\documentclass [11pt]{article}
\usepackage{graphicx}
\usepackage{color}
\usepackage{amsmath}
\usepackage{amssymb}
\usepackage{amscd}
\usepackage{amsthm}
\usepackage{amsopn}
\usepackage{xspace}
\usepackage{verbatim}
\usepackage{amsmath}
\usepackage{amsfonts}
\usepackage{epic}
\usepackage{eepic}
\usepackage[active]{srcltx}
\usepackage{empheq}
\definecolor{SeaGreen}{RGB}{46,139,87}
\definecolor{Navy}{RGB}{0,0,128}
\newtheorem{theorem}{Theorem}
\newtheorem{proposition}{Proposition}
\newtheorem*{acknowledgment*}{Acknowledgment}
\newtheorem{remark}{Remark}

\newcommand{\R}{\mathbb R}
\newcommand{\C}{\mathbb C}

\newcommand{\A}{\mathcal A}
\newcommand{\B}{\mathcal B}

\DeclareMathOperator{\Div}{div}

\def\XXint#1#2#3{{\setbox0=\hbox{$#1{#2#3}{\int}$ }
\vcenter{\hbox{$#2#3$ }}\kern-.6\wd0}}

\def\Dg {{\mathcal D}}

\newcommand{\supp}{\mathrm{supp}\,}

\def\Hg {{\mathcal H}}

\newcommand{\LL}{\mathcal L}

\newcommand{\OO}{\mathcal O}

\def\Ug {{\mathcal U}}
\def\Wg {{\mathcal W}}

\begin{document}
\bibliographystyle{ieeetr}

\date{}\title{Existence and stability of superconducting solutions for
the Ginzburg-Landau equations in the presence of weak electric currents}
  \author{Yaniv
  Almog\thanks{Department of Mathematics, Louisiana State University,
    Baton Rouge, LA 70803, USA}, Leonid Berlyand\thanks{Department of
    Mathematics, Pennsylvania State University, University Park, PA
    16802, USA}, Dmitry Golovaty\thanks{Department of Mathematics, The University of Akron, Akron, Ohio 44325,
    USA}, and Itai Shafrir\thanks{Department of Mathematics, Technion -
    Israel Institute of Technology, 32000 Haifa, Israel}}

\maketitle
\begin{abstract}
  For a reduced Ginzburg-Landau model in which the magnetic field is
  neglected,  we prove, for weak electric currents, the existence of a
  steady-state solution in a vicinity of the purely superconducting state. We
  further show that this solution is linearly stable.
\end{abstract}
\section{Introduction}
\label{sec:2}

Superconducting materials are characterized by a complete loss of
resistivity at temperatures below some critical threshold value. In
this state, electrical current can flow through a superconducting
sample while generating only a vanishingly small voltage drop. If the current
is increased above a certain critical level, however,
superconductivity is destroyed and the material reverts to the normal
state---even while it remains below the critical temperature.

In this work, we study this phenomenon within the framework of the
time-dependent Ginzburg-Landau model \cite{chhe03,goel68}, presented
here in a dimensionless form
   \begin{equation}
\label{eq:42}
\left\{
     \begin{aligned}
  &\frac{\partial u}{\partial t} +   i\phi u = \left(\nabla - iA \right)^{2} u + u 
  \left( 1 - |u|^{2} \right) & \text{ in } & \Omega\times\R_+ \,,\\
 -& \kappa^2\nabla \times \nabla \times A + \sigma \left(\frac{\partial A}{\partial t} + 
 \nabla\phi\right)  =   \Im\{\bar{u} \nabla u\}  + |u|^{2}A &
\text{ in } & \Omega\times\R_+\,, \\
 &(i\nabla+A)u\cdot\nu=0 \quad\text{and}\quad -\sigma\left(\frac{\partial A}{\partial t} + 
 \nabla\phi\right) \cdot\nu=J& \text{ on } & \partial\Omega\times\R_+ \,, \\
&u(x_,0)=u_0\quad\text{and}\quad  A(x_,0)=A_0& \text{ in } & \Omega\,, \\
     \end{aligned}
     \right.
\end{equation}
In the above system of equations, $u$ is the order parameter with
$|u|$ representing the number density of superconducting electrons.
Materials for which $|u| = 1$ are said to be purely
superconducting while those for which $u = 0$ are said to be in the
normal state.  We denote the magnetic vector potential by $A$---so
that the magnetic field is given by $h=\nabla\times A$---and by $\phi$ the
electric scalar potential. The constants $\kappa$ and $\sigma$ are the
Ginzburg-Landau parameter and normal conductivity, of the
superconducting material, respectively, and the quantity
$-\sigma(A_t+\nabla\phi)$ is the normal current. All lengths in \eqref{eq:42}
have been scaled with respect to the coherence length $\xi$ that
characterizes spatial variations in $u$. The domain $\Omega\subset\R^2$
occupied by the superconducting sample is separated from its exterior
by the boundary $\partial\Omega$ that consists of two parts,
$\partial\Omega=\partial\Omega_c\cup\partial\Omega_i$. Here $\partial\Omega_c$ corresponds to the portion of the
boundary through which current enters and exits the sample, while the
rest of the boundary, denoted by $\partial\Omega_i$, is electrically insulated.
The function $J:\partial\Omega\to\R$ with $\supp(J)=\partial\Omega_c$ represents the normal
current entering the sample. Note, that it is possible to prescribe
the electric potential on $\partial\Omega$ instead of the current.

Except for the initial conditions,
\eqref{eq:42} is invariant under the gauge transformation \cite{chhe03}
\begin{displaymath}
  A \to A+\nabla\omega \quad ; \quad u\to u e^{i\omega}  \quad ; \quad \phi\to \phi-
  \frac{\partial\omega}{\partial t}
\end{displaymath}
for some smooth $\omega$. Finally, one has to prescribe $h$ at a single
point on $\partial\Omega$ for all $t>0$ (cf. \cite{al12}).

It has been demonstrated in \cite{al12}, for a fixed current, that in the
limit $\kappa\to\infty$ one can formally obtain from \eqref{eq:42} the
following system of equations
\begin{equation}
\label{eq:49}
\left\{
  \begin{aligned}
&  \frac{\partial u}{\partial t} + i\phi u = \Delta u + u
\left( 1 - | u|^{2} \right) & \qquad &\text{in } \Omega\times\R_+, \\
&  \sigma\Delta\phi = \nabla\cdot [\Im(\bar{ u} \nabla u)] 
& \qquad & \text{in }  \Omega\times\R_+, \\
&  \frac{\partial u}{\partial\nu}=0 \text{ and }-\sigma\frac{\partial\phi}{\partial\nu}=J &\qquad &\text{on } \partial\Omega\times\R_+,  \\
&  u(x_,0)= u_0 & \qquad &\text{in } \Omega \,.
\end{aligned}
\right.
\end{equation}

The principal goal of the present paper is to study \eqref{eq:49} in the large domain limit. To this
end, we apply the transformation 
\begin{displaymath}
 t^1=\epsilon^2t \; ; \; x^1=\epsilon x\; ; \;  J^1= \frac{J}{\epsilon} \; ;\;  \phi^1=\frac{\phi}{\epsilon^2} \;
  ; \; \sigma^1=\sigma\epsilon^2 \,,
\end{displaymath}
to (\ref{eq:49}) and drop the superscript $1$ for notational convenience to obtain 
\begin{subequations}
\label{eq:1}
\begin{empheq}[left={\empheqlbrace}]{alignat=2}
&\frac{\partial u}{\partial t} + \LL_\epsilon  u =0 \qquad & \text{in }& \Omega\times\R_+, \\
&\sigma\Delta\phi = \nabla\cdot [\Im(\bar{ u} \nabla u)] \qquad  & \text{in }& \Omega\times\R_+, \\
&\frac{\partial u}{\partial\nu}=0 \text{ and }-\sigma\frac{\partial\phi}{\partial\nu}=J \qquad & \text{on }& \partial\Omega\times\R_+, \\
&u(x,0)= u_0 \qquad & \text{in }& \Omega \,.
\end{empheq}
\end{subequations}
Here
\begin{equation}
\label{eq:50}
 \LL_\epsilon u =  i\phi u - \Delta u - \frac{u}{\epsilon^2}
\left( 1 - | u|^{2} \right)\,.
\end{equation}
We assume in the sequel that $\Omega$ in (\ref{eq:1}) is independent of
$\epsilon$. 

Note that \eqref{eq:1} remains invariant under the transformation
\begin{displaymath}
u\to e^{i\omega(t)}u \quad ; \quad \phi\to\phi+\frac{\partial\omega}{\partial t} \,.
\end{displaymath}
We thus choose
\begin{displaymath}
  \omega=-\int_0^t\frac{(|u|^2\phi)_\Omega(\tau)}{(|u|^2)_\Omega(\tau)} \,d\tau \,,
\end{displaymath}
which guarantees that we have for all $t>0$,
\begin{equation}
\label{eq:2}
  (|u|^2\phi)_\Omega(t)\equiv0,
\end{equation}
where
\[(f)_\Omega:=\frac{1}{|\Omega|}\int_\Omega f\,dx.\]
 \begin{figure}[htp]
  \centering
\scalebox{0.55}{ \input{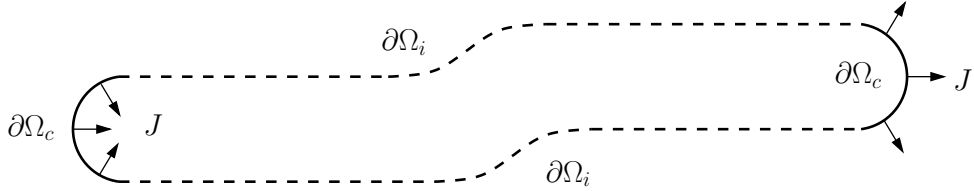}}
\caption{Schematics of a superconducting sample. The arrows denote the
  direction of the current flow.}
  \label{fig:1}
\end{figure}

The system of equations \eqref{eq:42} for a variety of domains and boundary conditions has
attracted significant interest among both physicists
\cite{ivko84}-\nocite{ivetal82,doetal98,voetal03}\cite{kabe14} and mathematicians \cite{al12}, \cite{ruetal07}-\nocite{rust08,ruetal10,ruetal10b}\cite{al08}.  A
different simplification of \eqref{eq:1} was derived by Du \& Gray
\cite{dugr96} for the same limit ($\kappa\to\infty$), but assuming that $J$ and $\sigma$ are of order 
$\OO(\kappa^2)$ (cf. \cite{duetal10}).

The focus of this work is mainly on the existence and stability of
steady-state solutions of \eqref{eq:1} for relatively small
currents. The main result that we prove is the following
\begin{theorem}
\label{thm:main}
Let $0<\delta_0$ and suppose that $\epsilon\|J\|_{H^{3/2}(\partial\Omega)}\leq\delta_0$. Then 
\begin{enumerate}
\item The system \eqref{eq:1} possesses a
steady-state solution $(u_s,\phi_s)\in H^2(\Omega,\C)\times H^2(\Omega,\R)$ whenever $\delta$ is sufficiently small.
Furthermore, there exists a constant $C(\Omega,\sigma)$, independent of both $\epsilon$ and $\delta_0$, such that
  \begin{displaymath}
    \|1-|u_s|\,\|_{2,2} \leq C\delta^2 \,.
  \end{displaymath}

\item The solution $(u_s,\phi_s)$ is linearly stable in the following sense. Given 
\begin{displaymath}
  \Ug= \{ u\in H^2(\Omega,\C) \, : \, \partial u/\partial\nu|_{\partial\Omega}=0 \,\} \,,
\end{displaymath} 
let
$\LL_\epsilon:\Ug\to L^2(\Omega,\C)$ be defined by \eqref{eq:50} where the potential $\phi(u)$ is assumed to solve 
\[
\left\{
\begin{aligned}
&\sigma\Delta\phi = \nabla\cdot [\Im(\bar{ u} \nabla u)] \qquad  & \text{in }& \Omega\times\R_+, \\
&-\sigma\frac{\partial\phi}{\partial\nu}=J \qquad & \text{on }& \partial\Omega\times\R_+, \\
&(|u|^2\phi)_\Omega(t)\equiv0. &&
\end{aligned}
\right.
\]
Then, there exist a $0\leq\delta_1\leq\delta$ such that, whenever
$\epsilon\|J\|_{H^{3/2}(\partial\Omega)}\leq\delta_1$, the semi-group associated with the
Fr\'{e}chet derivative $D \LL_\epsilon (u_s)$ is asymptotically stable.
\end{enumerate}
\end{theorem}
\begin{remark}
\normalfont Note that Theorem \ref{thm:main} is valid for every $0<\epsilon\leq1$ and in
particular when $\epsilon\ll\delta$.
\end{remark}
The rest of this paper is organized as follows. In the next
section, we prove the existence of steady-state solutions and
discuss their properties in
Theorem \ref{thm:stationary1}. The stability of these solutions (Proposition
\ref{prop:stable1}) is subsequently demonstrated in Section 3.

\section{Steady state solutions}
\label{sec:4}
In this section, we consider the steady-state solutions of
\eqref{eq:1}. Let $(u,\phi)$ denote a smooth solution of 
\begin{equation}
\label{eq:3}
\left\{
  \begin{aligned}
-&\Delta u + i\phi u = \frac{u}{\epsilon^2}\left( 1 - | u|^{2} \right) & \text{in }& \Omega\,, \\
&  \sigma\Delta\phi = \nabla\cdot [\Im(\bar{ u} \nabla u)] &  \text{in }&  \Omega\,, \\
&  \frac{\partial u}{\partial\nu}=0 \text{ and }-\sigma\frac{\partial\phi}{\partial\nu}=J & \text{on }& \partial\Omega\,,  \\
\end{aligned}
\right.
\end{equation}
If we set $u=\rho e^{i\chi}$ the problem takes the form
\begin{subequations}
\label{eq:4}
\begin{empheq}[left={\empheqlbrace}]{alignat=2}
-&\Delta\rho+\rho|\nabla\chi|^2 = \frac{\rho}{\epsilon^2}(1 - \rho^2) \qquad  & \text{in }& \Omega\,, \\
&\Div(\rho^2\nabla\chi)=\rho^2\phi \qquad & \text{in }& \Omega\,, \\
&  \sigma\Delta\phi = \Div(\rho^2\nabla\chi) \qquad  &  \text{in }&  \Omega\,, \\
&  \frac{\partial\rho}{\partial\nu}=\frac{\partial\chi}{\partial\nu}=0 \text{ and }-\sigma\frac{\partial\phi}{\partial\nu}=J \qquad & \text{on }& \partial\Omega\,,  \\
& (\chi)_\Omega=0\,. &&
\end{empheq}
\end{subequations}
In what follows, we assume that $J$ belongs to $H^{3/2}(\partial\Omega)$ and set
$\|J\|$ to be the $H^{3/2}$-norm of $J$.  Note that (\ref{eq:4}f) is
imposed in order to eliminate the degree of freedom that results from
the invariance of (\ref{eq:4}a-e) with respect to the transformation
$\chi\to\chi+C$, for any constant $C$. Moreover, any solution of
\eqref{eq:4} must satisfy
\begin{equation}
  \label{eq:41}
 \int_\Omega \rho^2\phi=0\,,
\end{equation}
as can be easily verified by integrating
(\ref{eq:4}b) and then using  (\ref{eq:4}d,e).

Assuming that current is sufficiently small, we seek an approximation
to the solution of \eqref{eq:4} that would be uniform in $\epsilon$ and, in
particular, would remain valid in the limit $\epsilon\to0$. To this end, we
fix the value of $\sigma$ while allowing for some dependence of $J$ on
$\epsilon$: a point that will be clarified in the sequel. The approximate
solution $(u,\phi)=(\rho_0e^{i\chi_0},\phi_0)$ when $\epsilon\ll1$ should satisfy
\begin{equation}
\label{eq:5}
\left\{
\begin{aligned}
  &\rho_0^2 = 1-\epsilon^2|\nabla\chi_0|^2 & \text{in }& \Omega\,, \\
   &-\sigma\Delta\phi_0 + \rho_0^2\phi_0 = 0 & \text{in }& \Omega\,, \\
    &\Div (\rho_0^2\nabla\chi_0) = \rho_0^2\phi_0  & \text{in }& \Omega\,, \\
    &\frac{\partial\phi_0}{\partial\nu} = -\frac{J}{\sigma} & \text{on }& \partial\Omega_c\,, \\
    &\frac{\partial\chi_0}{\partial\nu} = 0 & \text{on }& \partial\Omega_i\,, \\
    &(\chi_0)_\Omega =0 \,. &&
\end{aligned}
\right.
\end{equation}
Note that the only term dropped from \eqref{eq:4} to obtain
\eqref{eq:5} is $-\Delta\rho$ in (\ref{eq:4}a).  We first prove the existence
of solutions to \eqref{eq:5} for a sufficiently small 
\begin{displaymath}
  \delta:=\|J\|\epsilon.
\end{displaymath}
As will become clear later on, the solution of \eqref{eq:5},  whose
existence is proved below, serves a good approximation for a solution
of (\ref{eq:4}) whenver $\delta$ is sufficiently small, even if $\epsilon$ is
bounded away from zero.
\begin{proposition}
\label{prop:implicit}
 Let
 \begin{displaymath}
   \Hg_1 = \{ \phi\in H^3(\Omega)\,| \, \partial\phi/\partial\nu|_{\partial\Omega} = 0 \,\} \; ;
   \;  \Hg_2 = \{ \chi\in H^3(\Omega)\,| \, \partial \chi/\partial\nu|_{\partial\Omega} = 0,\  (\chi)_\Omega=0\} \,,
 \end{displaymath}
 and let $ \Wg_1 = \Hg_1\times\Hg_2$. There exist positive $\delta_0$
 and $C(\Omega,\sigma)$, such that the problem \eqref{eq:5} possesses a solution $(\chi_0,\phi_0)\in\Wg_1$
 satisfying
\begin{equation}
\label{eq:6}
  \|\chi_0\|_{3,2} + \|\phi_0\|_{3,2}  \leq C \|J\| \,,
\end{equation}
and
\begin{equation}
\label{eq:7}
  \|(1-\rho_0)\|_{2,2} \leq C\delta^2 \,,
\end{equation}
for all $0<\delta<\delta_0$ and
 $\sigma>0$.
\end{proposition}
\begin{proof}
  We make use of the implicit function theorem to prove the
  proposition. Denote by $\phi_{0,0}$ and $\chi_{0,0}$ the  solutions of 
\begin{subequations}
  \label{eq:8}
  \begin{equation}
  \begin{dcases}
    -\sigma\Delta\phi_{0,0} + \phi_{0,0} = 0 & \text{in } \Omega, \\
    \frac{\partial\phi_{0,0}}{\partial\nu} = -\frac{J}{\sigma} & \text{on } \partial\Omega \,,
  \end{dcases}
\end{equation}
and
\begin{equation}
  \begin{dcases}
    \Delta\chi_{0,0} = \phi_{0,0}  & \text{in } \Omega\,, \\[1.2ex]
    \frac{\partial\chi_{0,0}}{\partial\nu} = 0 & \text{on } \partial\Omega\,, \\[1.2ex]
    \left(\chi_{0,0}\right)_\Omega =0 \,,
  \end{dcases}
\end{equation}
\end{subequations}
respectively. For convenience we normalize the various fields by $\|J\|$ (we assume
$\|J\|>0$):
  \begin{equation}
\label{eq:9}
    \tilde{\chi}_0 =  \frac{\chi_0}{\|J\|} \quad ; \quad    \tilde{\phi}_0 =
    \frac{\phi_0}{\|J\|}\,,
  \end{equation}
and
  \begin{displaymath}
    \tilde{\chi}_{0,0} =  \frac{\chi_{0,0}}{\|J\|} \quad ; \quad    \tilde{\phi}_{0,0} =
    \frac{\phi_{0,0}}{\|J\|}\,.
  \end{displaymath}
Then we set
\begin{equation}
\label{eq:52}
  \tilde{\chi}_0=\tilde{\chi}_{0,0} + \omega_\delta \quad ; \quad
  \tilde{\phi}_0=\tilde{\phi}_{0,0} + \varphi_\delta \,.
\end{equation}

We begin by making the trivial observation that $ (\tilde{\phi}_{0,0}
,\tilde{\chi}_{0,0})\in\Wg_1$, i.e.,
\begin{equation}
\label{eq:10}
    \|\tilde{\chi}_{0,0}\|_{3,2} + \|\tilde{\phi}_{0,0}\|_{3,2} \leq C\,.
\end{equation}
We then define $F:\Wg_1\times \R \to H^1(\Omega,\R^2)$ by
\begin{displaymath}
  F(\varphi_\delta,\omega_\delta,\delta) =
  \begin{bmatrix}
    -\sigma\Delta\varphi_\delta + \varphi_\delta - \delta^2|\nabla\tilde{\chi}_0|^2\tilde{\phi}_0 \\[1.5ex]
    -\Delta\omega_\delta  + \varphi_\delta - \delta^2|\nabla\tilde{\chi}_0|^2\tilde{\phi}_0 + \delta^2\Div \big(|\nabla\tilde{\chi}_0|^2\nabla\tilde{\chi}_0\big)
  \end{bmatrix}\,.
\end{displaymath}
Note that (\ref{eq:52}) provides a one-to-one correspondence between the solutions of
$F(\varphi_\delta,\omega_\delta,\delta)=0$ and the solutions of \eqref{eq:5}.  It can be easily verified that $F$ is well-defined
because, by Sobolev embeddings,
\begin{multline*}
\||\nabla\tilde{\chi}_0|^2\tilde{\phi}_0 \|_{1,2} +  \|\Div \big(|\nabla\tilde{\chi}_0|^2\nabla\tilde{\chi}_0\big)\|_{1,2} \leq
 C\big( \big[\|\nabla\tilde{\chi}_0\|_\infty^2\|\tilde{\phi}_0\|_{1,2}\\
  +\|\nabla\tilde{\chi}_0\|_\infty\|\tilde{\chi}_0\|_{2,2}\|\tilde{\phi}_0\|_\infty +
  \|\nabla\tilde{\chi}_0\|_\infty^2\|\tilde{\chi}_0\|_{3 ,2}  + 
 \|\nabla\tilde{\chi}_0\|_\infty\|\tilde{\chi}_0\|_{2,4}^2\big] \\
\leq C  \big[\|\tilde{\chi}_0\|_{3 ,2}^3 +\|\tilde{\chi}_0\|_{3 ,2}^2
\|\tilde{\phi}_0\|_{2,2} \Big] \,.
\end{multline*}
Furthermore, we have $F(0,0,0)=0$. 

Let $\Dg:\Wg_1\times\R\to H^1(\Omega,\R^2)$ denote the linear operator
\begin{equation}
\label{eq:11}
  \Dg(\varphi,\omega,\delta) =   
\begin{bmatrix}
    \Dg_1(\varphi,\omega,\delta) \\
    \Dg_2(\varphi,\omega,\delta),
  \end{bmatrix}\,,
\end{equation}
where
\[\Dg_1(\varphi,\omega,\delta)=-\sigma\Delta\varphi + \varphi -2\delta^2\tilde{\phi}_0\nabla\tilde{\chi}_0\cdot\nabla\omega-\delta^2|\nabla\tilde{\chi}_0|^2\varphi\]
and
\begin{multline*}
\Dg_2(\varphi,\omega,\delta)=-\Delta\omega +\varphi -2\delta^2\tilde{\phi}_0\nabla\tilde{\chi}_0\cdot\nabla\omega \\ -\delta^2|\nabla\tilde{\chi}_0|^2\varphi+\delta^2\Div
    \big(|\nabla\tilde{\chi}_0|^2\nabla\omega +2(\nabla\tilde{\chi}_0\cdot\nabla\omega)\nabla\tilde{\chi}_0\big)
\end{multline*}
By using the same approach that we used to show that $F$ is
well-defined, it can be verified that
\begin{multline}
\label{eq:12}
  \|\Dg(\varphi,\omega,\delta)\|_{1,2} \leq C  \big[\delta^2\|\tilde{\chi}_0\|_{3
    ,2}^2\|\omega\|_{3 ,2}  + \|\varphi\|_{3 ,2} + \|\omega\|_{3 ,2} \\
+ \delta^2\|\tilde{\phi}_0\|_{3,2}\|\omega\|_{3 ,2}\|\tilde{\chi}_0\|_{3,2} + \delta^2\|\tilde{\chi}_0\|_{3,2}^2\|\varphi\|_{3 ,2} \big] \,.
\end{multline}
Similarly, we can demonstrate that for any
$(\tilde{\phi}_0,\tilde{\chi}_0)\in\Wg_1$ and $(\varphi,\omega)\in\Wg_1$ we have 
\begin{multline*}
  \|F(\varphi_\delta+\varphi,\omega_\delta+\omega,\delta)-F(\varphi_\delta,\omega_\delta,\delta)-\Dg(\varphi,\omega,\delta)\|_{1,2} \leq \\C\delta^2\big[\big(\|\tilde{\chi}_0\|_{3
    ,2}+ \|\tilde{\phi}_0\|_{3,2}\big)\|\omega\|_{3 ,2}^2 +
  \|\tilde{\chi}_0\|_{3,2}\|\omega\|_{3 ,2} \|\varphi\|_{3 ,2} + \|\omega\|_{3
    ,2}^2 \big( \|\omega\|_{3 ,2} + \|\varphi\|_{3 ,2}\big) \big]\,.
\end{multline*}
From the above inequality and \eqref{eq:12} it follows that $\Dg$ is
the Fr\'echet derivative of $F$ with respect to $(\varphi_\delta,\omega_\delta)$, and
that it is continuous in any neighborhood of $(0,0,0)$ in $\Wg_1\times\R$.

We can now conclude from \eqref{eq:11} that at $(0,0,0)$,
\begin{displaymath}
  DF(\varphi,\omega,0) =   
\begin{bmatrix}
    -\sigma\Delta\varphi + \varphi  \\
    -\Delta\omega +\varphi 
  \end{bmatrix}\,.
\end{displaymath}
It can be easily shown that $DF:\Wg_1\to H^1(\Omega,\R^2)$ is
invertible. 
Since $(-\sigma\Delta+1):\Hg_1\to H^1(\Omega)$ and $-\Delta:\Hg_2\to H^1(\Omega)$ are both
invertible, we have
\begin{displaymath}
  \|(DF)^{-1}\| \leq \|(-\sigma\Delta+1)^{-1}\| +  \|(-\Delta)^{-1})\| (1+\|(-\sigma\Delta+1)^{-1}\| )\,.
\end{displaymath}
(Note that $-\Delta$ is invertible since the average of $\omega$ in $\Omega$ must
vanish.)  Consequently, by the implicit function theorem
(cf. \cite{ni01}, for instance) we can find $\delta_0>0$ such that for every
$0<\delta<\delta_0$ there exists $(\varphi_\delta,\omega_\delta)\in\Wg_1$ for which
$F(\varphi_\delta,\omega_\delta,\delta)=0$. It readily follows that $(\varphi_\delta,\omega_\delta)$ converges in
$H^3(\Omega)\times H^3(\Omega)$ to $(0,0)$ as $\delta\to0$. In particular, we
obtain that 
\begin{equation}
\label{eq:13}
 \|\omega_\delta\|_{3,2}+ \|\varphi_\delta\|_{3,2}\xrightarrow[\delta\to0]{}0 \,.
\end{equation}
Combining the above  with \eqref{eq:10} completes the proof of
\eqref{eq:6}. The proof of \eqref{eq:7} follows as well, since
\begin{displaymath}
  \|(1-\rho_0)\|_{2,2}\leq C\delta^2\left(\|\nabla\tilde{\chi}_0\|_\infty \|
  \tilde{\chi}_0\|_{3,2} +  \|\tilde{\chi}_0\|_{2,4}^2\right) \,.
\end{displaymath}
\end{proof}
\begin{remark}
\normalfont  Note that $\delta_0$ is independent of $\epsilon$ because
  $F(\varphi_\delta,\omega_\delta,\delta)$ is independent of $\epsilon$ as well.
\end{remark}

Next, we show
\begin{theorem}
\label{thm:stationary1}
Let $\delta=\|J\|\epsilon$. There exists a $\delta_0>0$ such that \eqref{eq:3}
possesses a unique solution $(u,\phi)$ satisfying $\|u-u_0\|_{1,2}<\delta\epsilon$
for all $0<\epsilon\leq1$ and $0<\delta<\delta_0$. Furthermore, there exists a
constant $C=C(\Omega,\sigma)$, independent of both $\epsilon$ and $\delta$ such that
\begin{equation}
    \label{eq:14}
\|\rho-\rho_0\|_2 + \frac{1}{\|J\|}(\|\chi-\chi_0\|_2 + \|\phi-\phi_0\|_2) \leq
C \delta^2\epsilon^2 \,,
  \end{equation}
and
\begin{equation}
  \label{eq:15}
\|\rho-\rho_0\|_\infty + \|\rho-\rho_0\|_{1,2} + \epsilon\|\rho-\rho_0\|_{2,2}+ \frac{1}{\|J\|}(\|\chi-\chi_0\|_{2,2} +
\|\phi-\phi_0\|_{2,2}) \leq  C\delta^2\epsilon\,.
\end{equation}
Here $u=\rho e^{i\chi}.$
\end{theorem}
\begin{proof}
  We use Banach fixed point theorem in order to prove both existence
  and uniqueness. Recall the definition of $\tilde{\chi}_0$ and
  $\tilde{\phi}_0$ and set
  \begin{displaymath}
    \rho_1=\rho-\rho_0 \quad ; \quad \chi_1= \frac{1}{\|J\|}(\chi-\chi_0) \quad ; \quad  \phi_1=
    \frac{1}{\|J\|}(\phi-\phi_0)  \,.
  \end{displaymath}
It is easy to show that
  \eqref{eq:4} can be written in the equivalent form
  \[
  \left\{
  \begin{aligned}
    & \left(-\Delta+\|J\|^2|\nabla\tilde{\chi}_0|^2+\frac{1}{\epsilon^2}(3\rho_0^2-1)\right)\rho_1
      +2\|J\|^2\rho_0\nabla\tilde{\chi}_0\cdot\nabla\chi_1 = \Delta\rho_0 && \\ 
     &\qquad - \|J\|^2|\nabla\chi_1|^2\rho_0  -
      \|J\|^2\nabla\chi_1\cdot(2\nabla\tilde{\chi}_0 \nabla\chi_1)\rho_1-\frac{1}{\epsilon^2}(3\rho_0+\rho_1)\rho_1^2 & \text{in } & \Omega \\
    & -\Div(\rho_0^2\nabla\chi_1) - 2\Div \big(\rho_1\rho_0\nabla\tilde{\chi}_0\big)
      +\rho_0^2\phi_1 + 2\rho_0\rho_1\tilde{\phi}_0
      = 
      -\rho_1(2\rho_0+\rho_1)\phi_1 & & \\  
      & \qquad - \rho_1^2\tilde{\phi}_0 +\Div
      \big(\rho_1^2\nabla\tilde{\chi}_0\big)  + \Div \big(\rho_1(2\rho_0+\rho_1)\nabla\chi_1\big) & \text{in } & \Omega \\
     & -\sigma\Delta\phi_1 +\rho_0^2\phi_1 + 2\rho_0\tilde{\phi}_0\rho_1= -\rho_1(2\rho_0+\rho_1)\phi_1 - \rho_1^2\tilde{\phi}_0   & \text{in } & \Omega  \\
      & \frac{\partial\rho_1}{\partial\nu} =    \frac{\partial\chi_1}{\partial\nu} =
      \frac{\partial\phi_1}{\partial\nu}=0   & \text{on } & \partial\Omega     
    \end{aligned}
    \right.
    \]
We define the space
\begin{equation}
\label{eq:17}
  \Hg = \left\{ (\eta,\omega,\varphi)\in H^2(\Omega,\R^3)\,\big| \,
  (\nabla\eta,\nabla\omega,\nabla\varphi)\cdot\nu\big|_{\partial\Omega}=0 \,;\, (\omega)_\Omega=0
   \right\} ,
\end{equation}
and let $(\eta,\omega,\varphi)\in\Hg$ be the weak solution of the following
boundary value problem
\begin{subequations}
\label{eq:18}
\begin{empheq}[left={\empheqlbrace}]{alignat=2}
   &   -\left(\Delta-\|J\|^2|\nabla\tilde{\chi}_0|^2-\frac{1}{\epsilon^2}(3\rho_0^2-1)\right)\eta
    + 2\|J\|^2\rho_0\nabla\tilde{\chi}_0\cdot\nabla\omega  = f_1 & \quad
     \text{in } \Omega\;\, \\
&-\Div(\rho_0^2\nabla\omega) - 2\Div \big(\eta\rho_0\nabla\tilde{\chi}_0\big) +\rho_0^2\varphi +
2\rho_0\eta\tilde{\phi}_0= f_2
& \quad
     \text{in } \Omega\;\, \\
&   -\sigma\Delta\varphi +\rho_0^2\varphi + 2\rho_0\tilde{\phi}_0\eta = f_3 & \quad
     \text{in } \Omega \,,
\end{empheq} 
\end{subequations}
where $(f_1,f_2,f_3)\in L^2(\Omega,\R^3)$. 

{\em Step 1:} Prove that $v=(\eta,\omega,\varphi)$ is well defined. To this end we
use the Lax-Milgram lemma. Let
$w=(\tilde{\eta},\tilde{\omega},\tilde{\varphi})$. Define the bilinear form $B:\Hg\times\Hg\to\R$
\begin{multline}
\label{eq:19}
  B[v,w] = \langle\nabla\eta,\nabla\tilde{\eta}\rangle +
  \Big\langle\Big(\|J\|^2|\nabla\tilde{\chi}_0|^2+\frac{1}{\epsilon^2}(3\rho_0^2-1)\Big)\eta,\tilde{\eta}\Big\rangle \\ + 
  \|J\|^2\Big[2\langle\rho_0\nabla\omega,\tilde{\eta}\nabla\tilde{\chi}_0\rangle +
  \langle\rho_0\nabla\omega,\rho_0\nabla\tilde{\omega}\rangle+2\langle\rho_0\nabla\tilde{\omega},\eta\nabla\tilde{\chi}_0\rangle \\ +
  \langle\rho_0\tilde{\omega},(\rho_0\varphi+2\eta\tilde{\phi}_0)\rangle + C_0\big( \langle\sigma\nabla\varphi,\nabla\tilde{\varphi}\rangle +
  \langle\rho_0\tilde{\varphi},(\rho_0\varphi+2\eta\tilde{\phi}_0)\rangle\big)\Big]\,,
\end{multline}
where the (positive) value of $C_0$ will be determined later.
Since by \eqref{eq:6} both $\tilde{\chi}_0$ and $\tilde{\phi}_0$ are in
$H^3(\Omega)$, it readily follows from Sobolev embeddings
that there exists $C(\Omega,\epsilon)$ such that
\begin{displaymath}
  | B[v,w]| \leq C \|v\|_{1,2}\|w\|_{1,2} \,.
\end{displaymath}
To use the Lax-Milgram Lemma we need yet to consider the quadratic form
$B[v,v]$. Note that
\begin{equation*}
  B[v,v]=I_1+I_2+I_3
\end{equation*}
 where
 \begin{equation}
   \label{eq:20}
\begin{aligned}
   I_1&=\|\nabla\eta\|^2_2+\frac{1}{\epsilon^2}\|(3\rho_0^2-1)^{1/2}\eta\|_2^2\,,\\
 I_2&=\|J\|^2\left\{
 \|\rho_0\nabla \omega\|_2^2+4\langle\rho_0\eta\nabla\omega,\nabla\tilde{\chi}_0\rangle+\langle\rho_0\omega,(\rho_0\varphi+2\eta\tilde\phi_0)\rangle\right\}\,,\\
I_3&=\|J\|^2C_0\left\{ \sigma\|\nabla\varphi\|^2_2+\langle\rho_0\varphi,(\rho_0\varphi+2\eta\tilde\phi_0)\rangle\right\}\,.
 \end{aligned}
 \end{equation}

With the aid of the Cauchy-Schwarz inequality we obtain for any $\delta_1>0$,
\begin{multline}
  I_2\geq \|J\|^2\Big\{\frac{1}{2}
  \|\rho_0\nabla\omega\|^2_2-8\|\nabla\tilde{\chi}_0\|_\infty^2\|\eta\|^2_2-
 2\delta_1\|\rho_0\|_\infty^2\|\omega\|_2^2\\-(\frac{\|\rho_0\|_\infty^2}{4\delta_1})\|\varphi\|_2^2-
\frac{\|\tilde{\phi}_0\|_\infty^2}{\delta_1} \|\eta\|_2^2\Big\}\,.
\end{multline}
Applying Cauchy-Schwarz inequality to $I_3$ yields,
\begin{equation}
  \label{eq:21}
  I_3\geq \|J\|^2C_0\left\{\sigma\|\nabla\varphi\|_2^2+\frac{1}{2}\|\rho_0\varphi\|_2^2-2\|\tilde\phi_0\|_\infty^2\|\eta\|_2^2\right\}\,.
\end{equation}
By \eqref{eq:7} and Sobolev embeddings we have that
\begin{equation}
  \label{eq:22}
\|1-\rho_0\|_\infty\leq C\delta^2 ~\Longrightarrow~ \|(3\rho_0^2-1)-2\|_\infty\leq C\delta^2\,.
\end{equation}
Furthermore, by  the Poincare's inequality (recall that $(\omega)_\Omega=0$), there exists
a $\lambda=\lambda(\Omega)>0$ such that
\begin{displaymath}
  \|\nabla\omega\|_2^2 \geq \lambda\|\omega\|_2^2\,.
\end{displaymath}
Using \eqref{eq:20}--\eqref{eq:21}, \eqref{eq:22} and setting
$\delta_1=\lambda/8$ and $C_0=2/\delta_1$, we deduce the existence of $C(\Omega)$ such
that, for a sufficiently small $\delta$ the inequality 
\begin{equation}
\label{eq:23}
   | B[v,v]| \geq C\frac{\delta^2}{\epsilon^2}(\|\omega\|_{1,2}^2+\|\varphi\|_{1,2}^2) +
   \frac{1}{\epsilon^2}\|\eta\|_2^2 + \|\nabla\eta\|_2^2\,
\end{equation}
holds. We can thus conclude the existence of a unique $v\in\Hg$ such that
\begin{displaymath}
  B[v,w]=\langle F,w\rangle \quad \forall w\in\Hg \,,
\end{displaymath}
where  $F=(f_1,\|J\|^2f_2,C_0\|J\|^2f_3)$. 

Let $(\rho-\rho_0,\chi-\tilde{\chi}_0,\phi-\tilde{\phi}_0)\in\Hg$. We set
\begin{subequations}
\label{eq:24}
  \begin{align}
  & f_1 = \Delta\rho_0  - \|J\|^2|\nabla\chi_1|^2\rho_0  -
      \|J^2\|\nabla\chi_1\cdot(2\nabla\tilde{\chi}_0+ \nabla\chi_1)\rho_1-\frac{1}{\epsilon^2}(3\rho_0+\rho_1)\rho_1^2\\
  & f_2 = -\rho_1(2\rho_0+\rho_1)\phi_1 - \rho_1^2\tilde{\phi}_0 + \Div
      \big(\rho_1^2\nabla\tilde{\chi}_0\big)  + \Div \big(\rho_1(2\rho_0+\rho_1)\nabla\chi_1\big) \\
& f_3 = -\rho_1(2\rho_0+\rho_1)\phi_1 - \rho_1^2\tilde{\phi}_0        \,.
\end{align}
\end{subequations}
Substituting the above into \eqref{eq:18}, we can define the operator
$\A:\Hg\to \Hg$
\begin{displaymath}
  \A(\rho_1,\chi_1,\phi_1)=(\eta,\omega,\varphi)\,.
\end{displaymath}
We look for a fixed point of $\A$. 
 We equip $\Hg$ with the norm
 \begin{displaymath}
 \|w\|_\Hg = \|\eta\|_{1,2}+ \|\eta\|_\infty+\epsilon\|D^2\eta\|_2  + \|\omega\|_{2,2} +
 \|\varphi\|_{2,2} \,.
 \end{displaymath}

{\em Step 2:} Let $v=(\rho_1,\chi_1,\phi_1)$.  We prove that 
for sufficiently small $\epsilon$ and $\delta$
there exist $C(\Omega,\sigma)$ and $r(\epsilon,\delta)\leq C\delta\epsilon$ for which
\begin{equation}
\label{eq:25}
  v\in B(0,r) \Rightarrow A(v)\in B(0,r) \,.
\end{equation}
Let then $0<r\leq\epsilon$ and $v\in B(0,r)$. We begin by deriving a bound
on $\|\eta\|_2$ and $\|A(v)\|_{1,2}$. By (\ref{eq:24}a), \eqref{eq:6},
and \eqref{eq:22} we
have that
\begin{displaymath}
  \|f_1\|_2^2 \leq  C\Big\{\|\Delta\rho_0\|_2^2 + \|J\|^4\big[\|\nabla\chi_1\|_4^4+
  \|\rho_1\|_4^2\big(\|\nabla\chi_1\|_4^2+\|\nabla\chi_1\|_8^4\big)\big] +
  \frac{1}{\epsilon^4}\big(  \|\rho_1\|_4^4 + \|\rho_1\|_6^6 \big)\Big\} \,.
\end{displaymath}
Recall that, by \eqref{eq:7},
there exists $C(\Omega,\sigma)>0$, such that for a sufficiently small $\delta>0$ we
have 
\begin{displaymath}
  \|\Delta\rho_0\|_2 \leq C\delta^2 \,.
\end{displaymath}
Sobolev embeddings then yield
\begin{equation}
\label{eq:26}
    \|f_1\|_2 \leq C\Big( \delta^2 + \frac{r^2}{\epsilon^2}\Big)\,.
\end{equation}
Similarly, we obtain that
\begin{multline*}
|\langle\omega,f_2\rangle + \langle\varphi,f_3\rangle| \leq
C\Big\{\big[(\|\rho_1\|_4 +
\|\rho_1\|_8^2)\|\phi_1\|_4+
\|\rho_1\|_4^2\big](\|\omega\|_2+\|\varphi\|_2) \\ +
 \big[ \|\rho_1\|_4^2+(\|\rho_1\|_4 +
\|\rho_1\|_8^2)\|\nabla\chi_1\|_4\big] \|\nabla\omega\|_2 \Big\} \,,
\end{multline*}
and hence,
\begin{displaymath}
 |\langle\omega,f_2\rangle + \langle\varphi,f_3\rangle| \leq Cr^2(\|\omega\|_{1,2}+\|\varphi\|_2) \,. 
\end{displaymath}
Combining the above  with \eqref{eq:26} yields
\begin{equation}
  \label{eq:27}
|\langle\A(v),F\rangle| \leq C\Big[\Big( \delta^2 + \frac{r^2}{\epsilon^2}\Big)\|\eta\|_2 +
\frac{\delta^2r^2}{\epsilon^2}(\|\omega\|_{1,2}+\|\varphi\|_2)\Big]  \,. 
\end{equation}

As $B(\A(v),\A(v))=\langle\A(v),F\rangle$ we obtain by \eqref{eq:23} that
\begin{equation}
\label{eq:28}
  \|\eta\|_2 \leq C\left[\delta^2\epsilon^2+r^2\right] \,.
\end{equation}
Upon multiplying (\ref{eq:18}b) by $\omega$ and  (\ref{eq:18}c) by $\varphi$ we
sum the resulting equations and integrate over $\Omega$ to obtain
\begin{displaymath}
  \|\nabla\omega\|_2^2 + \sigma\|\nabla\varphi\|_2^2 + \|\varphi\|_2^2 \leq C\big[r^2(\|\omega\|_{1,2} +
  \|\varphi\|_2) + \|\eta\|_2(\|\omega\|_{1,2} + \|\varphi\|_2)\big]\,.
\end{displaymath}
Using Poincare's inequality we then obtain, with the aid of
\eqref{eq:28}, that
\begin{equation}
\label{eq:29}
  \|\omega\|_{1,2} + \|\varphi\|_{1,2} \leq C(r^2 + \delta^2\epsilon^2) \,.
\end{equation}
Substituting the above, together with \eqref{eq:28} into
\eqref{eq:27} and using \eqref{eq:23} yields
\begin{equation}
  \label{eq:30}
\|\nabla\eta\|_2\leq \frac{C}{\epsilon} (r^2 + \delta^2\epsilon^2) \,.
\end{equation}

To complete the proof of \eqref{eq:25} we rewrite first (\ref{eq:18}b) in
the form
\begin{displaymath}
  -\Div(\rho_0^2\nabla\omega) = 2\Div \big(\eta\rho_0\nabla\tilde{\chi}_0\big) -\rho_0^2\varphi -
2\rho_0\eta\tilde{\phi}_0+ f_2\,.
\end{displaymath}
We attempt to estimate the $L_2$-norm of the right-hand side. Clearly
\begin{equation}
\label{eq:31}
  \|2\Div \big(\eta\rho_0\nabla\tilde{\chi}_0\big)-\rho_0^2\varphi -
2\rho_0\eta\tilde{\phi}_0\|_2 \leq C (\|\eta\|_{1,2}+ \|\varphi\|_2)\leq \frac{C}{\epsilon} (r^2 + \delta^2\epsilon^2)\,.
\end{equation}
Furthermore, we have
\begin{equation}
  \label{eq:32}
\|\rho_1(2\rho_0+\rho_1)\phi_1 + \rho_1^2\tilde{\phi}_0\|_2 \leq C\big[(\|\rho_1\|_4 +
\|\rho_1\|_{8}^2)\|\phi_1\|_4 + \|\rho_1\|_{4}^2\big] \leq Cr^2
\end{equation}
and
\begin{multline*}
\|\Div\big(\rho_1^2\nabla\tilde{\chi}_0\big)\|_2 + \|\Div
\big(\rho_1(2\rho_0+\rho_1)\nabla\chi_1\big)\|_2\leq \\C\big(\|\rho_1\|_4^2 +
  \|\nabla\rho_1\|_2\|\rho_1\|_\infty+\|\rho_1\|_\infty\|\chi_1\|_{2,2}+\|\nabla\rho_1\|_2\|\nabla\chi_1\|_2 \big)\leq Cr^2 \,.
\end{multline*}
Combining the above with \eqref{eq:32} and \eqref{eq:31} yields for
$\epsilon\leq1$,
\begin{equation}
\label{eq:33}
  \|\Div(\rho_0^2\nabla\omega)\|_2 \leq \frac{C}{\epsilon} (r^2 + \delta^2\epsilon^2) \,.
\end{equation}
Writing
$\Div(\rho_0^2\nabla\omega)=\rho_0^2\Delta\omega+2\rho_0\nabla\rho_0\cdot\nabla\omega$,
using \eqref{eq:7} and standard elliptic estimates, we obtain from
\eqref{eq:33}  that
\begin{equation}
  \label{eq:34}
\|\omega\|_{2,2} \leq  \frac{C}{\epsilon} (r^2 + \delta^2\epsilon^2) \,.
\end{equation}
In a similar manner it is possible to show that
\begin{equation}
  \label{eq:35}
\|\varphi\|_{2,2} \leq  C(r^2 + \delta^2\epsilon^2) 
\end{equation}

To complete the proof we need yet to bound $\epsilon\|D^2\eta\|_2$
and $\|\eta\|_\infty$. To this end we
rewrite (\ref{eq:18}a) in the form
\begin{displaymath}
  -\Delta\eta = -\Big(\|J\|^2|\nabla\tilde{\chi}_0|^2-\frac{1}{\epsilon^2}(3\rho_0^2-1)\Big)\eta
    - 2\|J\|^2\rho_0\nabla\tilde{\chi}_0\cdot\nabla\omega  + f_1 \,.
\end{displaymath}
It easily follows that for sufficiently small $\delta$ we have
\begin{equation}
\label{eq:36}
  \Big\|
  \Big(\|J\|^2|\nabla\tilde{\chi}_0|^2-\frac{1}{\epsilon^2}(3\rho_0^2-1)\Big)\eta
  \Big\|_2 \leq  \frac{C}{\epsilon^2} (r^2 + \delta^2\epsilon^2) \,.
\end{equation}
Furthermore, as
\begin{displaymath}
  \|J\|^2\|\rho_0\nabla\tilde{\chi}_0\cdot\nabla\omega\|_2 \leq C\frac{\delta^2}{\epsilon^2}r^2 \,,
\end{displaymath}
we obtain with the aid of \eqref{eq:36} and \eqref{eq:26} that
\begin{equation}
\label{eq:37}
  \|\eta\|_{2,2} \leq  \frac{C}{\epsilon^2} (r^2 + \delta^2\epsilon^2) \,.
\end{equation}
It follows from Agmon's inequality (cf. \cite[Lemma 13.2]{ag65}) in
conjunction with \eqref{eq:28} that
\begin{equation}
\label{eq:38}
  \|\eta\|_\infty \leq C\|\eta\|_2^{1/2}\|\eta\|_{2,2}^{1/2}\leq \frac{C}{\epsilon} (r^2 + \delta^2\epsilon^2) \,.
\end{equation}
Combining the above with \eqref{eq:28}, \eqref{eq:30}, \eqref{eq:34},
\eqref{eq:37}, and \eqref{eq:35} yields
\begin{displaymath}
  \|\A(v)\|_\Hg \leq \frac{C}{\epsilon} (r^2 + \delta^2\epsilon^2) \,.
\end{displaymath}
We may thus choose $r=\delta\epsilon$ to obtain, for
a sufficiently small value of $\delta$, that
\begin{displaymath}
   \|\A(v)\|_\Hg\leq C\epsilon\delta^2 <r \,.
\end{displaymath}

{\em Step 3:} Let $(v_1,v_2)\in B(0,r)^2$. We prove that there exists
$\gamma<1$ such that
\begin{equation}
  \label{eq:39}
\|\A(v_1)-\A(v_2)\|_\Hg \leq \gamma\|v_1-v_2\|_\Hg \,.
\end{equation}
It can be easily verified that
\begin{align*}
  & \|f_1(v_1)-f_1(v_2)\|_2 \leq C
  \Big(\|J\|^2+\frac{1}{\epsilon^2}\Big)r\|v_1-v_2\|_\Hg \,, \\
  & \|f_2(v_1)-f_2(v_2)\|_2 \leq Cr\|v_1-v_2\|_\Hg \,,\\
& \|f_2(v_1)-f_2(v_2)\|_2 \leq  Cr\|v_1-v_2\|_\Hg \,.
\end{align*}
Let now $\A(v_1)=(\eta_1,\omega_1,\varphi_1)$ and $\A(v_2)=(\eta_2,\omega_2,\varphi_2)$. As
\begin{displaymath}
  B(\A(v_1)-\A(v_2),\A(v_1)-\A(v_2))=\langle\A(v_1)-\A(v_2),F(v_1)-F(v_2)\rangle \,, 
\end{displaymath}
we obtain  by \eqref{eq:23} that
\begin{displaymath}
  \|\eta_1-\eta_2\|_2 \leq Cr\|v_1-v_2\|_\Hg \,.
\end{displaymath}
The same procedure that led to \eqref{eq:29} and \eqref{eq:30} enables us to conclude that
\begin{displaymath}
  \|\omega_1-\omega_2\|_{1,2} + \|\varphi_1-\varphi_2\|_{1,2} \leq  Cr\|v_1-v_2\|_\Hg \,,
\end{displaymath}
and that
\begin{displaymath}
    \|\nabla(\eta_1-\eta_2)\|_2 \leq C\frac{r}{\epsilon}\|v_1-v_2\|_\Hg \,.
\end{displaymath}
We then proceed in precisely the same manner as in the derivation of
\eqref{eq:34} and \eqref{eq:35} to obtain that 
\begin{displaymath}
  \epsilon\|\omega_1-\omega_2\|_{2,2} +  \|\varphi_1-\varphi_2\|_{2,2} \leq  Cr\|v_1-v_2\|_\Hg \,.
\end{displaymath}
Finally, using the same procedure as in the derivation of
\eqref{eq:38} and \eqref{eq:37} we obtain that
\begin{displaymath}
      \|\eta_1-\eta_2\|_\infty +\epsilon^2\|\eta_1-\eta_2\|_{2,2} \leq C\frac{r}{\epsilon}\|v_1-v_2\|_\Hg \,.
\end{displaymath}
Combining all of the above then yields
\begin{displaymath}
  \|\A(v_1)-\A(v_2)\|_\Hg\leq C\frac{r}{\epsilon}\|v_1-v_2\|_\Hg\,,
\end{displaymath}
and since $r=\delta\epsilon$, we obtain \eqref{eq:39} for a sufficiently small
value of $\delta$. 
\end{proof}

\section{Linear stability}
\label{sec:satble}
In what follows, we examine the linear stability of the solution we
have obtained in the previous section. To this end, let
\begin{displaymath}
  \Ug= \{ u\in H^2(\Omega,\C) \, :\, \partial u/\partial\nu|_{\partial\Omega}=0 \,\} 
\end{displaymath}
and define the
non-linear operator $\LL_\epsilon:\Ug\to L^2(\Omega,\C)$ by
\begin{displaymath}
  \LL_\epsilon  u = -\Delta u +i\phi u - \frac{u}{\epsilon^2}(1-|u|^2)\,.
\end{displaymath}
for any $u\in \Ug$. In the above $\phi$ denotes a non-local, non-linear operator of $u$.
We define $\phi$, in view of \eqref{eq:2},  as the solution of
\begin{displaymath}
  \begin{dcases}
    \sigma\Delta\phi =\Div (\Im\{\bar{u}\nabla u\})  & \text{in } \Omega \\
\frac{\partial\phi}{\partial\nu} = -\frac{J}{\sigma}  & \text{on
} \partial\Omega \\
\left(|u|^2\phi\right)_\Omega=0\,. &
  \end{dcases}
\end{displaymath}
The system \eqref{eq:1} can then be written in the form
\begin{equation}
\label{eq:112}
  u_t+\LL_\epsilon u=0\,.
\end{equation}

We look for the spectrum of $\A=D\LL_\epsilon(u_s)$---the Fr\'echet derivative
of $\LL_\epsilon$ at $u_s$. Set $\phi_s=\phi(u_s)$. It can be readily verified that
\begin{equation}
\label{eq:113}
  \A u = -\Delta u +i(\phi_su+\hat{\varphi} u_s)-\frac{u}{\epsilon^2}(1-\rho_s^2) + \frac{2u_s}{\epsilon^2}
  \Re(\bar{u}_su) \,,
\end{equation}
where $\hat{\varphi}(u,u_s)$ is a non-local linear operator given by the solution of
\begin{equation}
\label{eq:114}
    \begin{dcases}
    \sigma\Delta\hat{\varphi}  =\Div (\Im\{\bar{u}_s\nabla u +\bar{u}\nabla u_s\})  & \text{in } \Omega \\
\frac{\partial\hat{\varphi}}{\partial\nu} = 0  & \text{on
} \partial\Omega \\
\left(|u_s|^2\hat{\varphi}+2\phi_s\Re(\bar{u}_su)]\right)_\Omega=0\,. &
  \end{dcases}
\end{equation}
Note that $\A$ has a non-trivial kernel, i.e., $\A\,iu_s=0$. This
non-trivial kernel reflects the fact that
\begin{displaymath}
  e^{-i\Theta}\LL_\epsilon e^{i\Theta}=\LL_\epsilon \,. 
\end{displaymath}
Let $\{u_n\}_{n=0}^\infty$ denote the system of eigenfunctions associated
with $\A$ where $u_0=iu_s$. By Theorem 16.5 in \cite{ag65} we
have ${\rm span }\{u_n\}_{n=0}^\infty=L^2(\Omega,\C)$ and we   can  thus set
\begin{displaymath}
  D(\A) = \Ug\cap\overline{{\rm span \,}\{u_n\}_{n=1}^\infty}
\end{displaymath}
as the domain of $\A$, thereby eliminating $iu_s$ from the domain. 

Let
\begin{equation*}
  \tilde{u}=\tilde{\rho}e^{i\tilde{\chi}}=(\rho_s+\delta^\prime\rho)e^{i(\chi_s+\delta^\prime\chi)}\in\Ug
\end{equation*}
 denote an infinitesimal
perturbation of $u_s$, where $\delta^\prime$ is a small parameter. Then, $\tilde{u}=u_s+\delta'u+o(\delta')$ where 
\begin{equation}
\label{eq:40}
    u=e^{i\chi_s}(\rho+i\rho_s\chi)\,.
\end{equation}
Consider then the linear operator
\begin{displaymath}
  \B=e^{-i\chi_s}\A e^{i\chi_s}\,,
\end{displaymath}
defined on $D(\B)=e^{-i\chi_s}D(\A)$. More explicitly, we have
\begin{displaymath}
  D(\B) = \{ v\in\Ug \,|\,
\end{displaymath}
Since $e^{i\chi_s}$ is a
unitary operator, we have $\sigma(\A)=\sigma(\B)$. We write any  $v\in D(\B)$ as
$v=e^{-i\chi_s}u$ with $u\in D(\A)$. Substituting into \eqref{eq:113} yields 
\begin{multline} 
\label{eq:115}
  \B v= -\Delta v + |\nabla\chi_s|^2v -\frac{1}{\epsilon^2}(1-\rho_s^2)v +
  \frac{2}{\epsilon^2}\rho_s^2\Re v \\+i\Big((-\Delta\chi_s)v-2\nabla\chi_s\cdot\nabla v +  \phi_sv + \rho_s\varphi(v)\Big) \,,
\end{multline}
where 
$\varphi=\varphi(v)$ is given (according to \eqref{eq:114}) by the solution of
\begin{equation}
\label{eq:116}
\begin{dcases}
    \sigma\Delta\varphi  =\Div \big(\Im(\rho_s\nabla v+{\bar v}\nabla\rho_s)+2\rho_s\nabla\chi_s\Re v\big)  & \text{in } \Omega \\
\frac{\partial\varphi}{\partial\nu} = 0  & \text{on
} \partial\Omega \\
\left(\rho_s^2\varphi+2\phi_s\rho_s\Re v\right)_\Omega=0\,. &
  \end{dcases}
\end{equation}
 In view of \eqref{eq:40} we have $v=\rho+i\rho_s\chi$.
Next we look for a non-trivial solution to the eigenvalue problem $\B
v=\lambda v$, i.e.,
\begin{subequations}
\label{eq:117}
  \begin{gather}
    \Re(\B v)=\Re(\lambda v) \\
    \Im(\B v)=\Im(\lambda v) \,.
  \end{gather}
\end{subequations}

We now prove the stability of the solution of \eqref{eq:3} in the
neighborhood of $(\rho_0e^{i\chi_0},\phi_0)$ where $(\rho_0,\chi_0,\phi_0)$ are
given by \eqref{eq:5}. We establish this for a sufficiently small
value of $\|J\|\epsilon$, which is precisely the limit where existence has been
obtained in the previous section.

\begin{proposition}
  \label{prop:stable1}
  Let $u_s$ denote a solution of \eqref{eq:1} in the neighborhood of
  $(\rho_0e^{i\chi_0},\phi_0)$ given by \eqref{eq:5}. Furthermore, let
  $\A=D\LL_\epsilon(u_s):D(\A)\to L^2(\Omega,\C)$. There exists $\delta_0>0$, such
  that for all $0<\epsilon\leq1$ and $0<\delta<\delta_0$ we have
  \begin{equation}
    \label{eq:121}
\min_{\lambda\in\sigma(\A)}\Re\lambda>0\,,
  \end{equation}
where $\delta=\|J\|\epsilon$.
\end{proposition}
\begin{proof}
Let $v=\rho+i\rho_s\chi \in D(\B)$ denote an eigenfunction of $\B$ associated
with the eigenvalue $\lambda$. By  \eqref{eq:115} and \eqref{eq:117}, the triplet $(\rho,\chi,\lambda)$ must satisfy
the following problem
\begin{subequations}
\label{eq:122}
\begin{empheq}[left={\empheqlbrace}]{alignat=2}
  &   -\Delta\rho + \rho|\nabla\chi_s|^2 + 2\rho_s\nabla\chi\cdot\nabla\chi_s -
  \frac{\rho}{\epsilon^2}(1-3\rho_s^2) = \lambda_r\rho-\lambda_i\rho_s\chi \quad &
     \text{in }& \Omega \\
   &  -\Div (\rho_s^2\nabla\chi) - 2\Div (\rho_s\rho\nabla\chi_s) + \rho_s^2\varphi + 2\rho_s\rho\phi_s
     = \lambda_r\rho_s^2\chi+\lambda_i\rho_s\rho \quad &   \text{in }& \Omega \\
&-\sigma\Delta\varphi +\Div (\rho_s^2\nabla\chi) + 2\Div (\rho_s\rho\nabla\chi_s)= 0 \quad  & \text{in
}& \Omega  \\
& \frac{\partial\rho}{\partial\nu}=\frac{\partial\chi}{\partial\nu} = \frac{\partial\varphi}{\partial\nu} = 0 \quad &  \text{on
}& \partial\Omega \\
&\left(\rho_s^2\varphi+2\phi_s\rho_s\rho\right)_\Omega=0 \,, &&
  \end{empheq}
\end{subequations}
where $\lambda_r=\Re\lambda$ and $\lambda_i=\Im\lambda$.  Note that to obtain
(\ref{eq:122}a) we need to use (\ref{eq:4}b). Since the spectrum of $\B$ is
discrete (cf. \cite{ag65} chapter 15), it suffices to show that all
critical values of $\lambda$, for which non-trivial solutions for the above
problem exist (excluding, of course, $u=iu_s$), lie in the right hand
side of $\C$.

Taking the inner product
in $L^2(\Omega)$ of (\ref{eq:122}c) with $\sigma\varphi-\chi$ we obtain 
\begin{displaymath}
  \|\nabla(\sigma\varphi-\chi)\|_2^2 = \langle(\rho_s^2-1)\nabla\chi,\nabla(\sigma\varphi-\chi)\rangle  +
  2\langle\rho_s\nabla(\sigma\varphi-\chi),\rho\nabla\chi_s\rangle  \,.
\end{displaymath}
It can be easily demonstrated using \eqref{eq:6}, \eqref{eq:15}, and
Sobolev embeddings that
\begin{equation}
  \label{eq:123}
\|\rho\nabla\chi_s\|_2 \leq \|\nabla\chi_0\|_\infty\|\rho\|_2+\|\nabla(\chi_s-\chi_0)\|_p\|\rho\|_q \leq
C\|J\|(\|\rho\|_2 + \delta^2\epsilon\|\rho\|_{1,2})\,.
\end{equation}
In the above $p>2$ and $q=2p/(p-2)$.
Consequently, using \eqref{eq:7}
and Sobolev embeddings and recalling that $\|J\|=\delta/\epsilon$, it readily follows that
\begin{equation}
  \label{eq:124}
\|\nabla(\sigma\varphi-\chi)\|_2 \leq C\big(\delta^2\|\nabla\chi\|_2+\delta(\|\rho\|_2/\epsilon + \delta^2\|\rho\|_{1,2})\big)  \,.
\end{equation}
Next, we multiply (\ref{eq:122}b) by $\chi$ and (\ref{eq:122}a) by $\rho$,
then integrate their sum by parts to obtain
\begin{multline}
\label{eq:46}
  \|\rho_s\nabla\chi\|_2^2 + \|\nabla\rho\|_2^2 + 4\langle\rho_s\nabla\chi,\rho\nabla\chi_s\rangle +
  \|\rho\nabla\chi_s\|_2^2 + \frac{2}{\epsilon^2}\|\rho\|_2^2\\+
  \langle\rho_s\chi,\rho_s\varphi+2\rho\phi_s\rangle \leq \frac{3}{\epsilon^2}\|1-\rho_s^2\|_\infty\|\rho\|_2^2
  +\lambda_r(\|\rho_s\chi\|_2^2+\|\rho\|_2^2) \,.
\end{multline}
We now write,
\begin{align*}
 \int_\Omega \rho_s\chi(\rho_s\varphi+2\rho\phi_s) =\int_\Omega \rho_s(\chi-\sigma\varphi) (\rho_s\varphi+2\rho\phi_s)+\int_\Omega
 \sigma\rho_s\varphi(\rho_s\varphi+2\rho\phi_s)\\=\int_\Omega \rho_s(\chi-\sigma\varphi)
 (\rho_s\varphi+2\rho\phi_s)+\sigma\|\rho_s\varphi+2\rho\phi_s\|_2^2-2\sigma\int_\Omega \rho\phi_s (\rho_s\varphi+2\rho\phi_s)\,,
\end{align*}
which together with (\ref{eq:122}e) yields
\begin{multline}
\label{eq:43}
   \int_\Omega \rho_s\chi(\rho_s\varphi+2\rho\phi_s)\geq \sigma\|\rho_s\varphi+2\rho\phi_s\|_2^2\\-
   \big(\|\rho_s[\sigma\phi-\chi-(\sigma\varphi-\chi)_\Omega]\|_2+2\sigma\|\rho\phi_s\|_2\big)\|\rho_s\varphi+2\rho\phi_s\|_2
\end{multline}
Applying Cauchy's inequality, $2ab\leq2\alpha a^2+\frac{2}{\alpha}b^2$, with
$\alpha=\frac{1}{4}$, leads to 
\begin{equation}
  \label{eq:44}
 4\langle\rho_s\nabla\chi,\rho\nabla\chi_s\rangle\geq -\frac{1}{2}\|\rho_s\nabla\chi\|_2^2-8\|\rho\nabla\chi_s\|_2^2\,.
\end{equation}
 Recall that by \eqref{eq:6} and \eqref{eq:14} we have that $\|\nabla\chi_s\|_2\leq
 C\delta/\epsilon$, and hence we obtain from \eqref{eq:44} that
 \begin{equation}
   \label{eq:45}
 4\langle\rho_s\nabla\chi,\rho\nabla\chi_s\rangle\geq -\frac{1}{2}\|\rho_s\nabla\chi\|_2^2-\frac{C\delta^2}{\epsilon^2}\,.
 \end{equation}
Furthermore, because $\|1-\rho_s\|_\infty\leq C\delta^2$ by \eqref{eq:7} and \eqref{eq:15},
we substitute \eqref{eq:43} and \eqref{eq:45} into \eqref{eq:46} to
obtain 
\begin{multline}
\label{eq:47}
    \|\nabla\rho\|_2^2 + \frac{1}{2}\|\nabla\chi\|_2^2 + \sigma\|\rho_s\varphi+2\rho\phi_s\|_2^2 +
    \frac{2}{\epsilon^2}\|\rho\|_2^2 \\ \leq C\delta^2\Big[\frac{1}{\epsilon^2}\|\rho\|_2^2
    +\|\nabla\rho\|_2^2 + \|\nabla\chi\|_2^2\Big] \\+
    \big(\|\rho_s[\sigma\phi-\chi-(\sigma\varphi-\chi)_\Omega]\|_2\\+2\sigma\|\phi_s\rho\|_2\big)\|\rho_s\varphi+2\rho\phi_s\|_2+\lambda_r(\|\rho_s\chi\|_2^2+\|\rho\|_2^2)\,.
\end{multline}
By \eqref{eq:6} and \eqref{eq:14} we have also $\|\phi_s\|_\infty\leq C\delta/\epsilon$, and thus,
from \eqref{eq:124} and Poincar\'e inequality we get 
\begin{multline}
  \label{eq:48}
  \big(\|\rho_s[\sigma\phi-\chi-(\sigma\varphi-\chi)_\Omega]\|_2+2\sigma\|\phi_s\rho\|_2\big)\|\rho_s\varphi+2\rho\phi_s\|_2\leq\\
  \frac{\sigma}{2}\|\rho_s\varphi+2\rho\phi_s\|_2^2+C\delta^2(\frac{1}{\epsilon^2}\|\rho\|_2^2+\|\nabla\chi\|_2^2+\delta^2\|\nabla\rho\|_2^2)\,. 
\end{multline}
Finally, by \eqref{eq:47} and \eqref{eq:48} we have that
\begin{multline}
\label{eq:125}
  \|\nabla\rho\|_2^2 + \frac{1}{2}\|\nabla\chi\|_2^2 +\frac{\sigma}{2}\|\rho_s\varphi+2\rho\phi_s\|_2^2  +
    \frac{2}{\epsilon^2}\|\rho\|_2^2 \leq  C\delta^2\Big[\frac{1}{\epsilon^2}\|\rho\|_2^2 \\
    +\|\nabla\rho\|_2^2 + \|\nabla\chi\|_2^2+ \|\rho_s\varphi+2\rho\phi_s\|_2^2 \Big] +\lambda_r(\|\rho_s\chi\|_2^2+\|\rho\|_2^2) \,.
\end{multline}
For a sufficiently small $\delta$ we obtain that $\lambda_r\geq0$. If $\lambda_r=0$ we
have, for a sufficiently small value of $\delta$, that
\begin{displaymath}
  \|\rho\|_2^2 + \|\nabla\chi\|_2^2 =0 \,.
\end{displaymath}
Consequently, any eigenfunction associated with any eigenvalue on the
imaginary axis (for which $\lambda_r=0$) must be of the form $v=Ki\rho_s$ where
$K\in\R$ is a constant. From the definition of $D(\B)$ we easily
conclude that $K=0$. Hence, $\lambda_r>0$.
\end{proof}
\begin{acknowledgment*} {\rm This research was supported by US-Israel BSF grant
    no.~2010194.}
\end{acknowledgment*}

\bibliography{potential}
\end{document}